\documentclass[11pt]{article}
\usepackage{amssymb, latexsym, mathtools, amsthm}
\begingroup
    \makeatletter
    \@for\theoremstyle:=definition,remark,plain\do{%
        \expandafter\g@addto@macro\csname th@\theoremstyle\endcsname{%
            \addtolength\thm@preskip\parskip
            }%
        }
\endgroup
\usepackage{enumerate}
\usepackage{pgf,tikz}
\usepackage{pdfsync}
\usepackage{txfonts}
\usepackage[T1]{fontenc}
\usepackage{booktabs}
 \usepackage{graphicx}
\definecolor{dnrbl}{rgb}{0,0,0.3}
\definecolor{dnrgr}{rgb}{0,0.3,0}
\definecolor{dnrre}{rgb}{0.5,0,0}
\usepackage[colorlinks=true, citecolor=dnrgr, linkcolor=dnrre, urlcolor=dnrre] {hyperref}
\usepackage{xcolor}
\usepackage{parskip}
\usepackage{tabularx, colortbl}

\usepackage{geometry}
 \usepackage[scriptsize, up]{caption}
\usepackage{pgf,tikz}
\usetikzlibrary{decorations, decorations.pathmorphing}
\usetikzlibrary {shapes}

\theoremstyle{plain}
\newtheorem{thm}{Theorem}[section]

\newtheorem{prop}[thm]{Proposition}
\newtheorem{lem}[thm]{Lemma}
\newtheorem{coro}[thm]{Corollary}
\newtheorem{defi}[thm]{Definition}
\numberwithin{equation}{subsection}
\makeatletter

\let\c@table\c@figure
\makeatother

\newcommand{\restr}{\upharpoonright}  
\newcommand{\de}{\downarrow} 

\newcommand{\bigo}[1]{\mathop{\bf O}\/\left({#1}\right)}

\newcommand{\NST}{Nies, Stephan and Terwijn\ }

\newcommand{\ml}{Martin-L\"{o}f }

\newcommand{\eg}{e.g.\ }
\newcommand{\ie}{i.e.\ }
\newcommand{\ce}{c.e.\ }

\newcommand{\pf}{prefix-free }

\renewenvironment{abstract}
 { \normalsize
  \list{}{
    \setlength{\leftmargin}{.0cm}%
    \setlength{\rightmargin}{\leftmargin}%
    }%
  \item {\bf \abstractname.} \relax}
 {\endlist}

\newcommand{\pfm}{prefix-free machine }
\newcommand{\pfmn}{prefix-free machine}
\newcommand{\KG}{Ku\v{c}era-G\'{a}cs }

\usepackage{titlesec}
 \titleformat{\paragraph}
{\normalfont\normalsize\em\bfseries}{\theparagraph}{1em}{$\blacktriangleright$\hspace{0.2cm}}
\titlespacing*{\paragraph}
{0pt}{3.25ex plus 1ex minus .2ex}{0.5ex plus .2ex}

\title{Pointed computations and \ml randomness\thanks{{\em Date:} \today. Barmpalias was supported by the 
1000 Talents Program for Young Scholars from the Chinese Government, grant no.\ D1101130.
Additional support was received by
the Chinese Academy of Sciences (CAS) and the Institute of Software of the CAS.
Lewis-Pye was supported by a Royal Society University 
Research Fellowship.
Angsheng Li was partially supported by the 
National Basic Science Program (973 Program Group) entitled 
Cyber Information and Computational Theory on Big Data with Grant No. 2014CB340302.}}
\author{George Barmpalias  \and Andrew Lewis-Pye \and Angsheng Li}
\date{{\small\em For Barry, for whom the magnificent incomputability of the world was a deeply held belief.\\ The only response? An attempt at understanding this chaos at a higher order.}}
\begin{document}
\maketitle
\begin{abstract}
Schnorr showed that a real $X$ is \ml random if and only if $K(X\restr_n)\geq n-c$ 
for some constant $c$ and all $n$, where $K$ denotes the \pf complexity function. Fortnow (unpublished) and \NST
\cite{MR2140044} observed that the condition $K(X\restr_n)\geq n-c$ can be replaced with
$K(X\restr_{r_n})\geq r_n-c$, for any
fixed increasing computable sequence $(r_n)$, in this characterization.
The purpose of this note is to establish  
the following generalisation of this fact. 
We show that $X$ is \ml random
if and only if $\exists c\ \forall n\ K(X\restr_{r_n})\geq r_n-c$, where $(r_n)$ is any fixed 
{\em pointedly $X$-computable} sequence, in the sense that $r_n$ is computable from
$X$ in a self-delimiting way, so that at most the first $r_n$ bits of $X$ are queried in the computation
of $r_n$. 
On the other hand,  we also show that there are reals $X$ 
which are very far from being \ml random, 
but for which there exists some
$X$-computable sequence $(r_n)$ such that  $\forall n\ K(X\restr_{r_n})\geq r_n$.  
\end{abstract}
\vfill
\noindent{\bf George Barmpalias}\\[0.1em]
\noindent State Key Lab of Computer Science, 
Inst.\ of Software, Chinese Acad.\ Sci., Beijing, China.
School of Mathematics and Statistics, 
Victoria University of Wellington, New Zealand. \\
\textit{E-mail:} \texttt{\textcolor{dnrgr}{barmpalias@gmail.com}}.
\textit{Web:} \texttt{\href{http://barmpalias.net}{http://barmpalias.net}}
\vfill
\noindent{\bf Andrew Lewis-Pye}\\[0.1em]  
\noindent Department of Mathematics,
Columbia House, London School of Economics, 
Houghton St., London, WC2A 2AE, United Kingdom. \\
\textit{E-mail:} \texttt{\textcolor{dnrgr}{A.Lewis7@lse.ac.uk.}}
\textit{Web:} \texttt{\textcolor{dnrre}{http://aemlewis.co.uk}} 
\vfill
\noindent{\bf Angsheng Li}\\[0.1em]
\noindent State Key Lab of Computer Science, 
Inst.\ of Software, Chinese Acad.\ Sci., Beijing, China. \\
\textit{E-mail:} \texttt{\textcolor{dnrgr}{angsheng@ios.ac.cn.}}
 \vfill\thispagestyle{empty}
\clearpage
\newcommand{\upfm}{universal prefix-free machine }
\newcommand{\pfn}{prefix-free}
\newcommand{\upfmn}{universal prefix-free machine}
\newcommand{\lhs}{left-hand-side }
\newcommand{\rhs}{right-hand-side }
\newcommand{\nua}{\nu_{\ast}}
\newcommand{\sz}{$\Sigma^0_1$ }
\newcommand{\BHMN}{Bienvenu, H{\"{o}}lzl, Miller and Nies }
\newcommand{\BGKNT}{ Bienvenu, Greenberg, Ku\v{c}era, Nies and Turetsky }
\newcommand{\CDFT}{Chalcraft, Dougherty, Freiling, and Teutsch }
\newcommand{\BDM}{Barmpalias, Downey, and McInerney }

\section{Introduction}
A well known result of Schnorr (see Chaitin \cite{MR0411829}) is that Martin-L\"{o}f's notion of
algorithmic randomness from \cite{MR0223179} can be expressed in terms of incompressibility with
respect to \pf machines. In particular, a real $X$ is \ml random if and only if $\exists c\forall n\ K(X\restr_n)>n-c$,
where $K$ denotes the \pf complexity function. The latter condition says that there exists a constant
$c$ such that all the initial segments of $X$ are {\em $c$-incompressible} (in a \pf sense).
As reported in Downey and Hirschfeldt \cite[Proposition 6.1.4]{rodenisbook}, 
Fortnow (unpublished) and \NST \cite{MR2140044}
showed that Schnorr's characterisation remains valid if we replace the condition
$\exists c\forall n\ K(X\restr_n)>n-c$ with 
$\exists c\forall n\ K(X\restr_{r_n})>r_n-c$, where $(r_n)$ is any computable increasing sequence.

In this note we consider the extent to which this fact can be generalised to incomputable
increasing sequences $(r_n)$. It is well known that there are reals which are not \ml random, yet 
 have infinitely many $0$-incompressible initial segments. Hence this characterisation does not hold
for arbitrary increasing sequences $(r_n)$. We consider the case when $(r_n)$ is computable from $X$.
Our main result is that if $(r_n)$ is computable from $X$ in a certain restricted way, then
$X$ is \ml random if and only if $\exists c\forall n\ K(X\restr_{r_n})>r_n-c$. 
On the other hand we show how to construct reals $X$ 
and $X$-computable sequences $(r_n)$ such that
the above equivalence fails, so we have
$\forall n\ K(X\restr_{r_n})\geq r_n$ but $X$ is very far from being \ml random.

\subsection{Pointed computability}
Our main result is that if $X$ computes $(r_n)$ in a certain natural fashion, then
$\exists c\forall n\ K(X\restr_{r_n})>r_n-c$ is a  sufficient and necessary condition for the \ml randomness of $X$.
In this section we formalise the notion of oracle-computability required in order for this equivalence to hold, which we call pointed computability.

A Turing functional $\Phi$ can be thought of as a machine which takes as inputs a number $n$ and 
a program $\sigma$, and either halts on these inputs producing a number $\Phi^{\sigma}(n)$ as output,   or else diverges.
The consistency of $\Phi$ requires that if $\rho_0\subseteq\rho_1$ and $\Phi^{\rho_0}(n)\de$, then the computation $\Phi^{\rho_1}(n)$ is identical to that for 
$\Phi^{\rho_0}(n)$, yielding the same output.
Then $\Phi^X(n)$ can be defined as $\lim_s \Phi^{X\restr_s}(n)$. 
Without loss of generality, given a Turing functional $\Phi$, a string $\rho$ and a number $n$, we may assume that
$\Phi^{\rho}(n)\de$ implies 
$|\rho|\geq n$ and $|\rho|\geq \Phi^{\rho}(n)$.
This is a standard convention and it is not hard to see that if $\Psi^X=Z$ for two reals and a Turing functional 
$\Psi$ which might not obey the convention, then there exists a Turing functional $\Phi$ which does obey the stated convention and for which 
$\Phi^X=Z$.

Given a Turing functional $\Phi$, a number $n$ and strings $\rho_0,\rho_1$, the consistency property says that
if $\Phi^{\rho_0}(n)\neq\Phi^{\rho_1}(n)$ then we must have that $\rho_0\ |\ \rho_1$,
\ie the finite oracles differ at a digit which is less than
$\min\{|\rho_0|,|\rho_1|\}$.
The following definition is based on a more stringent consistency requirement.

\begin{defi}[Pointed computations]\label{GmcyBzTi6F}
Given a real $X$, we say that a sequence 
$(r_n)$ is pointedly $X$-computable  if $(r_n)$ is (strictly) increasing and there exists a Turing functional $\Phi$ such that $\Phi^{X \upharpoonright_{r_n}}(n)\downarrow =r_n$ for each $n$. 
\end{defi}
The latter condition in Definition \ref{GmcyBzTi6F} says
that some oracle Turing machine computes each $r_n$ from $X$ with oracle-use bounded above by $r_n$.
Note that, without loss of generality, we may assume that the oracle-use in this computation is exactly $r_n$.

As a typical example of pointed computations (suppressing the monotonicity requirement for now), 
consider an oracle machine
which starts on input $n$ by reading increasingly longer initial segments of 
the oracle tape, and eventually stops after $s>n$ steps, with output $s$.
If, for a given $X$, such a machine converges for every $n$, then it produces
a pointed computation in the sense of Definition \ref{GmcyBzTi6F}.
Another example is the settling time of a non-computable \ce set $A$: let $r_0=0$ and for each
$n$ let $r_{n+1}$ be the least number which is larger than $r_n$ and such that 
$A_{s}\restr_n=A\restr_n$ for all $s\geq r_{n+1}$.  Then $(r_i)$ is pointedly $A$-computable and
non-computable.

Later we will note that there are weaker notions of computability which suffice for our
characterization of \ml randomness. One such notion is the condition that
$K(r_n\ |\ \tau)=\bigo{1}$ for all $n$
and all $\tau\supseteq X\restr_{r_n}$. 
Note that the latter is the non-uniform version of the notion of
Definition \ref{GmcyBzTi6F}.

\subsection{Our results}
Our main result is the following, which we prove in Section \ref{7ZZctiXHJr}.
\begin{thm}[Randomness condition]\label{fSVU93INe}
Suppose that $(r_n)$ is pointedly computable from a real $X$. 
Then $X$ is \ml random if and only if
$\exists c\ \forall n\ K(X\restr_{r_n})>r_n-c$.
\end{thm}

It is well-known that there are reals which are not \ml random, yet
 have infinitely many incompressible initial segments.
Hence Theorem \ref{fSVU93INe}
does not hold if we simply waive the requirement that $(r_n)$ is pointedly computable from $X$.
One may ask, however,  if
Theorem \ref{fSVU93INe} continues to hold if we merely require that 
$(r_n)$ is computable from $X$ and not that it is pointedly $X$-computable.
It is not surprising that the latter question has a negative answer. 
One way to exhibit an example witnessing this fact, is to construct a real
with infinitely many incompressible initial segments, which computes the halting problem
and is not \ml random.
Since the \pf complexity function is computable from the halting problem
$\emptyset'$, given such an oracle $X$ we can effectively find infinitely many  $t$ such that $K(X\restr_t)\geq t-c$.
This gives the following fact.

\begin{prop}\label{NRuMPE7kEF}
Suppose that $X$ computes the halting problem, $X$ is not \ml random and there
exists some constant $c$ and  infinitely many $n$ such that $K(X\restr_n)\geq n-c$. 
Then $X$ computes an increasing sequence $(r_n)$ such that $K(X\restr_{r_n})\geq r_n-c$ for all $n$.
\end{prop}

In Section \ref{J6R5qSplUM} we present two ways of constructing oracles $X$ which have the properties 
mentioned in 
Proposition \ref{NRuMPE7kEF}, thus establishing the following.

\begin{thm}\label{qZcue1XzUo}
There exists a real $X$ 
and an $X$-computable increasing sequence $(r_n)$, such that 
$r_n<K(X\restr_{r_n})$ for all $n$, and $X$ is not \ml random.
\end{thm}

Our first construction of such $X$ involves starting from a \ml random $Y$ which computes the
halting problem, and inserting zeros at certain places, thus causing $X$ to be non-random, while preserving its
ability to calculate lengths at which its initial segments have high \pf complexity.
The second construction of such an oracle $X$ is more flexible, and gives  a real which is
highly non-random, in the sense that its characteristic sequence has a computable subsequence of zeros.

\subsection{Related concepts and results from the literature}
A central notion studied in this paper is that of a set $X$ which is able to compute a sequence of positions in
its binary expansion where the corresponding initial segments are incompressible. Clearly
every \ml random has this property, but there are also many reals with this property
which are not \ml random. This notion might remind some readers of the {\em autocomplex reals}
(see \cite{Kjos-HanssenMS06,Kjos-HanssenMStrans} or \cite[Section 8.16]{rodenisbook}), which
are the reals $X$ which compute a non-decreasing unbounded function $f$ such that
$K(X\restr_n)\geq f(n)$ for all $n$.
Moreover, Theorem \ref{fSVU93INe} has  similarities to a result by
Miller and Yu \cite{milleryutran2}, which says that if $\sum_i 2^{-g(i)}<\infty$ and $g$ is
computable from $X$ with identity oracle-use, then $\exists c\ \forall n\ K(X\restr_n)\leq n+ g(n)+c$. 
Note that the latter result can be seen as an extension of the following consequence of the
Kraft-Chaitin  inequality:
\begin{equation}\label{OtnW7AMnku}
\parbox{12cm}{if $g$ is computable and $\sum_i 2^{-g(i)}<\infty$  then
there exists a constant $c$ such that for all $X$ and all $n$ we have
$K(X\restr_n)\leq n+ g(n)+c$.}
\end{equation}
It is clear that the result of 
Miller and Yu \cite{milleryutran2} is related to its special case 
\eqref{OtnW7AMnku} in the same way as
our Theorem \ref{fSVU93INe} is related to results  originally obtained by
Fortnow (unpublished) and \NST
\cite{MR2140044}, discussed earlier.

\section{Proof of Theorem \ref{fSVU93INe}}\label{7ZZctiXHJr}
We need to show that 
if $(r_n)$ is pointedly $X$-computable, 
then $X$ is \ml random if and only if
$\exists c\ \forall n\ K(X\restr_{r_n})>r_n-c$.

The `only if' direction in this statement is trivial, so it remains to show the `if' direction. 
We prove the contrapositive. Assuming that $X$ is not \ml random
and $(r_n)$ is pointedly $X$-computable, we show that for each constant $c$
there exists some $n$ such that $K(X\restr_{r_n})\leq r_n-c$.
Let $\Phi$ be
a Turing functional such that $\Phi^X(n)=r_n$
for all $n$, and such that
for every $n,r, \rho$ such that $\Phi^{\rho}(n)\de=r$ we have
$\Phi^{\rho\restr_r}(n)\de$,
$|\rho|>n$, $|\rho|\geq r$ and $\Phi^{\rho}(i)\de$ for all $i<n$.
Let $U$ be the underlying optimal universal \pf machine.

We define a \pfm $M$ which makes use of the descriptions of $U$ as follows.
For each vector $(\rho,\sigma,n,t)$, such that: 
\begin{itemize}
\item $\Phi^{\rho}(i)\de$ for all $i\leq n+1$ and $\Phi^{\rho}(i)<\Phi^{\rho}(i+1)$ for all $i< n+1$; 
\item $t\in (\Phi^{\rho}(n),\Phi^{\rho}(n+1)]$ and $U(\sigma)=\rho\restr_t$;   
\end{itemize}
let $\tau$ be the digits of $\rho$ between
digit $t$ and digit  $\Phi^{\rho}(n+1)$ and 
note that
\[
\parbox{12cm}{ $U(\sigma)\ast\tau=\rho\restr_{\Phi^{\rho}(n+1)}$,
$\Phi^{U(\sigma)\ast\tau}(n+1)=\Phi^{\rho}(n+1)$ and  $\Phi^{U(\sigma)\ast\tau}(n)=\Phi^{\rho}(n)$.
}
\]
Let $M$ describe $\rho\restr_{\Phi^{\rho}(n+1)}$ with the string $\sigma\ast\tau$, defining
$M(\sigma\ast\tau)=\rho\restr_{\Phi^{\rho}(n+1)}$.
This completes the definition of $M$.

It follows immediately from the definition that  $M$ is effectively calculable. Next we show that $M$ does not allocate two different strings the same
description. Given two identical $M$-descriptions  $\sigma_0\ast\tau_0=\sigma_1\ast\tau_1$,
since $U$ is a \pfm we must have $\sigma_0=\sigma_1$ and $\tau_0=\tau_1$.
By the construction of $M$, the string described in both cases is then $U(\sigma_0) \ast \tau_0$. 
In a similar manner,  we may show that $M$ is a \pfmn.  Suppose that 
$\sigma_0\ast\tau_0 \subseteq \sigma_1\ast\tau_1$ are two descriptions issued by $M$,
resulting from the vectors  $(\rho_0,\sigma_0,n_0,t_0)$, $(\rho_1,\sigma_1,n_1,t_1)$ respectively. 
Since $U$ is a \pfm and $\sigma_0,\sigma_1$ are $U$-descriptions, we have $\sigma_0=\sigma_1$,
$t_0=t_1$ and  
$\tau_0\subseteq \tau_1$. 
Let $\sigma$ be $\sigma_0$ and let $t$ be $t_0$.
Then $\rho_0\subseteq \rho_1$ and hence\footnote{An intuitive description of the argument that follows
for showing that $\tau_0=\tau_1$ is this: each of $\rho_0,\rho_1$ determine a sequence 
$\Phi^{\rho_0}(i), i\leq n_0+1$, although $\rho_1$ may be able to define $\Phi^{\rho_1}(i)$ for $i>n_0+1$. 
However $t\in (\Phi^{\rho_0}(n_0), \Phi^{\rho_0}(n_0+1)]$ so by the construction of $M$ we have that both
$\tau_0, \tau_1$ equal the digits of $\rho_0$ (or equivalently $\rho_1$) between digit $t$ and digit 
$\Phi^{\rho_0}(n_0+1)$.}
\begin{equation}\label{v7af1UQRYd}
\Phi^{\rho_0}(n_0+1)=\Phi^{U(\sigma)\ast\tau_0}(n_0+1)=\Phi^{U(\sigma)\ast\tau_1}(n_0+1)=\Phi^{\rho_1}(n_0+1). 
\end{equation}
Thus 
$(\Phi^{\rho_0}(n_0), \Phi^{\rho_0}(n_0+1)]= (\Phi^{\rho_1}(n_0), \Phi^{\rho_1}(n_0+1)]$, and $t$ belongs to both intervals. This means that $n_0=n_1$, 
because otherwise
$t_1$ would have to belong to a different interval, not the one determined by $n_0$ and $\rho_1$,
since the values of $\Phi$ are strictly monotone. This would be a contradiction by the choice of
$(\rho_1, \sigma_1, n_1, t_1)$ in the definition of $M$ and the fact that $t_0=t_1$ which was established earlier.
So let $n=n_0=n_1$.
Since $U(\sigma)\ast\tau_0=\rho_0\restr_{\Phi^{\rho_0}(n+1)}$ and
$U(\sigma)\ast\tau_1=\rho_1\restr_{\Phi^{\rho_1}(n_0+1)}$, by \eqref{v7af1UQRYd}
the strings
$U(\sigma)\ast\tau_0$, $U(\sigma)\ast\tau_1$ have the same length, so $\tau_0=\tau_1$,
which shows that the two descriptions $\sigma_0\ast\tau_0$,  $\sigma_1\ast\tau_1$
are identical. This completes the proof that $M$ is a \pfmn.

It remains to show that if $X$ is not \ml random, then for each $c$ there exists some $n$
with $K(X\restr_{r_n})\leq r_n-c$.  
Let $d$ be a constant such that $K(\eta)\leq K_M(\eta)+d$ for all strings $\eta$.
Given any constant $c$, since $X$ is not \ml random,
there exists some $t>0$ such that  $K(X\restr_{t})\leq t-c-d$. Let $n$ be such that
$t\in (r_n, r_{n+1}]$. Then $M$ will describe $X\restr_{r_{n+1}}$
with $\sigma\ast\tau$ where $\sigma$ is the shortest description of 
$X\restr_{t}$ and the length of $\tau$ is $r_{n+1}-t$. The length of $\sigma$ is
$K(X\restr_{t})$, which is at most $t-c-d$. We have: 
\[
K_M(X\restr_{r_{n+1}})\leq (t-c-d) + (r_{n+1}-t)=r_{n+1}-c-d.
\]
So $K(X\restr_{r_{n+1}})\leq r_{n+1}-c$, as required.

\section{Proof of Theorem \ref{qZcue1XzUo}}\label{J6R5qSplUM}
As discussed in the introduction, we present two different constructions
of a real $X$ which computes an increasing sequence $(r_n)$ such that
$K(X\restr_{r_n})>r_n$ for all $n$ and $X$ is not \ml random.

\subsection{An ad hoc construction}
One way to construct a real $X$ with the property of Proposition \ref{NRuMPE7kEF} is to start
from a a \ml random real $Y$ which computes the halting problem, and insert zeros
into $Y$ in a way that does not change the fact that $\emptyset'$ is computable from the resulting oracle, but does ensure non-randomness. Recall that  a \ml random real $Y$ which computes the halting problem
exists by the \KG theorem  \cite{MR820784,MR859105}.
In this construction we use the result of Chaitin \cite{MR0411829}, which asserts that: 
\begin{equation}\label{2qm7AKqzzj}
\parbox{10cm}{if $Z$ is \ml random then $\lim_s \big(K(Z\restr_s)-s\big)=\infty$.}
\end{equation}
We also use the fact that for each 
real $Z$ which is \ml random and each string $\sigma$, the real $\sigma\ast Z$ is \ml random, and the fact from
\cite{chaitinincomp} that:
\begin{equation}\label{6aaKXae54g}
\parbox{11cm}{if $Z$ is \ml random and $f$ is a partial computable function on strings, then
if $f(Z\restr_n)\de$ for infinitely many $n$, there are infinitely many $t$ such that
$f(Z\restr_t)\de\neq Z(t)$.}
\end{equation}
The reader may observe that a 
partial computable prediction rule $f$ as above which is always successful on $Z$ would give rise to
a computable martingale which succeeds on $Z$, which we know is not possible for \ml random reals
(\eg see \cite[Section 6.3]{rodenisbook}). 

Let $\Phi$ be a functional via which $Y$ computes the complexity function $K$, i.e.\ such that $\Phi^Y(\sigma)=K(\sigma)$ for all $\sigma$. We form $X$ from $Y$ by inserting 0s at various positions, in a stage by stage process. The real $X$ is defined as the limit of a sequence $X_s$ and we form each $X_{s+1}$ from $X_s$ by inserting a 0 at position $t_{s+1}$, which means that we define $X_{s+1}(n)=X_s(n)$ for $n<t_{s+1}$, $X_{s+1}(n)=0 $ for $n=t_{s+1}$ and $X_{s+1}(n+1)=X_s(n)$ for $n\geq t_{s+1}$.  
At stage 0 we define $X_0=Y$, and (for convenience) define $t_0=-1$. Now inductively suppose that we have performed stages $0,\dots,k$, and that we have recorded $t_0,\dots,t_k$. For any string $\tau\subset X_{k}$, let $\tau^{\ast}$ be the string which results from removing all of the 0s that we have inserted during the stages $\leq k$. 
At step $k+1$ we search for $\sigma \subset X_k$ of length  $>t_{k}+2$ and $\tau \subset X_{k}$ with $\sigma \subset \tau$ such that $\Phi^{\tau^{\ast}}$ 
 computes $K(\sigma)$ and 
$K(\sigma)>|\sigma|$. By \eqref{2qm7AKqzzj}, it follows that such $\sigma$ and $\tau$ exist. 
Then we define $t_{k+1}=|\tau|$ and insert a 0 at position $t_{k+1}$.

This completes the construction of $X$ given $Y$. From the  construction it follows that $X$ computes both $Y$ and
the sequence $(t_k)$. This follows because $X$ is able to retrace the construction. Inductively suppose that $X$ has been able to retrace the construction up until the end of stage $k$, and so knows the values $t_0,\dots,t_k$. Then, using the oracle for $X$ we can perform the same search that was carried out at stage $k+1$, but using $X$ rather than $X_k$:  we search for $\sigma \subset X$ of length  $>t_{k}+2$ and $\tau \subset X$ with $\sigma \subset \tau$ such that $\Phi^{\tau^{\ast}}$ 
 computes $K(\sigma)$ and 
$K(\sigma)>|\sigma|$. Since the next zero is inserted after $\tau$, the result of the search is the same as when $X_k$ was used during the construction at stage $k+1$. Then $t_{k+1}=|\tau|$, completing the induction step. Therefore $X$ computes the halting set $\emptyset'$. 
Moreover there is a partial computable function $f$ such that for each $\sigma\subset X$ we have $f(\sigma)\de$
if and only if $\sigma=X\restr_{t_k}$ for some $k$ ($f$ uses $\sigma$ as an oracle to try and retrace the construction and converges on $\sigma$ if it is of length $t_k$ for some $k$ in the retraced construction). For each $k$, the next digit $t_k$ of $X$ is a 0, so $f$ is a partial computable prediction rule that succeeds on $X$,
which means that $X$ is not \ml random.
This completes the construction of a set $X$
with the properties of Proposition \ref{NRuMPE7kEF}.

\subsection{A refined construction}
Here we construct the required $X$ by finite extensions.
This construction can be combined with other requirements.
For example, $X$ can be highly non-random, in the sense that it has a computable sequence of 0s.
We need some facts from the theory of \pf Kolmogorov complexity.
For each string $\sigma$, let $\sigma^{\ast}$ denote the shortest \pf description of $\sigma$
(if there are many shortest descriptions, we consider the one which describes $\sigma$ first).
Also let $K(\tau\ |\ \rho)$ denote the \pf complexity of $\tau$ relative to string $\rho$.
The following is a relativized version of Chaitin's counting theorem from \cite{MR0411829}.

\begin{lem}[Relativized counting theorem]\label{TvutaENxZq}
There exists a constant $c$ such that for all $\sigma,r$ and all $n>\sigma$:
\[
\Big|\Big\{\tau\ \big|\ \sigma\subseteq\tau\wedge\tau\in 2^n \wedge 
K(\tau\ |\ \sigma^{\ast})\leq |\tau|-r-K(\sigma)\Big\}\Big|\leq
2^{n-K(\sigma)+c-r-K(n\ |\ \sigma^{\ast})}.
\]

\end{lem}
\begin{proof}
Given $\sigma$ we define $F(n\ |\ \sigma^{\ast})$  for $n>|\sigma|$ 
to be the $-\log$ of the weight of the \pf descriptions relative to $\sigma^{\ast}$ 
which describe extensions of $\sigma$ of length $n$. Then
since the relative \pf complexity $K$ is a minimal information measure, 
there exists a constant $c$ such that for all $n,\sigma$,
\begin{equation}\label{iNOCVFzuTZ}
2^{-F(n\ |\ \sigma^{\ast})}< 2^{-K(n\ |\ \sigma^{\ast})+c}.
\end{equation}
We claim that the constant $c$ has the property of the statement of the lemma.
For a contradiction, suppose that this is not the case. Then for some $n$ there are more than
$2^{n-K(\sigma)+c-r-K(n\ |\ \sigma^{\ast})}$ many
$\tau$ such that 
$K(\tau\ |\ \sigma^{\ast})\leq |\tau|-r-K(\sigma)$.
In that case we have 
\[
2^{-F(n\ |\ \sigma)}> 2^{n-K(\sigma)+c-r-K(n\ |\ \sigma^{\ast})}\cdot 2^{-n+r+K(\sigma)}=
2^{c-K(n\ |\ \sigma^{\ast})},
\]
which contradicts
\eqref{iNOCVFzuTZ}. This contradiction concludes the proof of the lemma.
\end{proof}
Recall the symmetry of information fact from \cite{Gacs:74,MR0411829}:
\[
K(\tau)+K(\sigma\ |\ \tau^{\ast})=K(\sigma)+K(\tau\ |\ \sigma^{\ast})+\Omega(1)
\]
where $f=g+\Omega(1)$ for two functions $f,g$ means that 
$|f(n)-g(n)|$ is bounded above.
Since $K(\tau)=K(\sigma,\tau)+\bigo{1}$ for all strings $\sigma,\tau$, 
by the symmetry of information, Lemma \ref{TvutaENxZq} has the following corollary.
\begin{coro}[Relativized counting, again]\label{DvBDInu2O3}
There exists a constant $c$ such that 
\[
\Big|\Big\{\tau\ \big|\ \sigma\subseteq\tau\wedge\tau\in 2^n \wedge 
K(\tau)\leq |\tau|-r\Big\}\Big|\leq
2^{n-K(\sigma)+c-r-K(\sigma\ |\ \tau^{\ast})-K(n\ |\ \sigma^{\ast})}
\]
for all $\sigma,r$ and all $n>\sigma$.
\end{coro}
Corollary \ref{DvBDInu2O3} is the tool we are going to use for our finite extension construction.
The problem we face is, given a string $\sigma$ to find an extension $\tau$ such that
$K(\tau)>|\tau|$. 
Since there are $2^{n-|\sigma|}$ many extensions of $\sigma$ of length $n$,
by Corollary \ref{DvBDInu2O3}
it suffices to consider $n$ such that
$2^{n-K(\sigma)+c-K(n\ |\ \sigma^{\ast})}<2^{n-|\sigma|}$,
which means that
$|\sigma|+c<K(\sigma)+K(n\ |\ \sigma^{\ast})$ so
\[
K(n\ |\ \sigma^{\ast})>|\sigma|-K(\sigma)+c.
\]
Such a number $n$ clearly exists in the interval $[|\sigma|, |\sigma|+2^{c+|\sigma|-K(\sigma)}]$.
The quantity $|\sigma|-K(\sigma)$ is sometimes called the {\em randomness deficiency} of $\sigma$.
We have shown that:
\begin{equation}\label{Shff9gJ5ND}
\parbox{12cm}{There exists a constant $c$ such that each $\sigma$ can be extended by less than
$2^{c+|\sigma|}-1$ many bits to a string $\tau$ with $K(\tau)>|\tau|$.}
\end{equation}
We are ready to construct the required real $X$ by finite extensions 
\begin{equation}\label{4nfpY6ZwAn}
\rho_0\subset \tau_1\subset\rho_1\subset\tau_2\subset\cdots
\end{equation}
In this construction the lengths $\ell_i:=|\tau_i|$ will be computable while the strings $\rho_i$ will be
chosen so that $K(\rho_i)>|\rho_i|$ for all $i$.
For each $i$ the string $\tau_{i+1}$ will be the concatenation of $\rho_i$ with a string 
$10\dots 0$ such that $|\tau_i|=\ell_i$. So for each $i$ the string $\rho_i$ will be uniformly computable
from $\tau_i$: simply find the first 1 starting from position $|\tau_i|$ in $\tau_i$ and moving to the left,
and if this 1 is at position $t$ then $\rho_i=\tau_i\restr_{t}$. Let $c$ be the constant in \eqref{Shff9gJ5ND}.

Let $\rho_0$ be the empty string $\lambda$, 
so that $K(\rho_0)>|\rho_0|$.
Let $\tau_1$ be the string $\rho_0\ast 10$ and let
$\ell_1=|\rho_0|+2=2$.
Note that the last digit of $\tau_1$ is a 0.
By  \eqref{Shff9gJ5ND}
there exists an extension $\rho_1$ of $\tau_1$
such that $|\rho_1|-|\tau_1|< 2^{c+\ell_1}-1$ and $K(\rho_1)>|\rho_1|$.
Then let $\tau_2$ be the concatenation of $\rho_1$ with a string $10\dots 0$
such that the length $\ell_2:=|\tau_2|$ is $\ell_1+2^{c+\ell_1}$.
Note that $\tau_2$ is longer than $\rho_1$ by at least 2 bits, so
the last digit of $\tau_2$ is a 0.
Similarly, we can choose an extension $\rho_2$ of $\tau_2$ of 
less than $2^{c+\ell_2}-1$ many additional bits such that
$K(\rho_2)>|\rho_2|$.
As before we let $\tau_3$ be the concatenation of $\rho_2$ with a string $10\dots 0$
such that the length of $\tau_3$ is 
$\ell_3:=\ell_2+2^{c+\ell_2}$. 
Note that $\tau_3$ is longer than $\rho_2$ by at least 2 bits, so the
last digit of $\tau_3$ is a 0.

The construction continues similarly, thus defining the computable sequence of lengths
$\ell_{n+1}=\ell_n+2^{c+\ell_n}$ for each $n>0$, where $\ell_1=2$,
and the sequences
\eqref{4nfpY6ZwAn} such that for all $i>0$
we have $K(\rho_i)>|\rho_i|$, $|\tau_i|=\ell_i$ and the last digit of $\tau_i$ is 0.
If we let $X$ be the infinite extensions of the strings \eqref{4nfpY6ZwAn} then we have $X(\ell_i-1)=0$ for all $i$,
so $X$ is not \ml random. On the other hand $X\restr_{\ell_i}=\tau_i$ for all $i$, and since $\tau_i$ uniformly
computes $\rho_i$, we have that the sequence $(|\rho_i|)$ is  $X$-computable.
Finally $K(X\restr_{|\rho_i|})=K(\rho_i|)>|\rho_i|$ for all $i$, which concludes the verification of the 
required properties of $X$.

\section{Conclusion and discussion}
We have generalized the criterion for \ml randomness 
by Fortnow (unpublished) and \NST
\cite{MR2140044},
which says that
if $X$ has incompressible segments of a computable sequence of lengths $(r_n)$, then it is \ml random.
We proved that the condition that $(r_n)$ is computable can be replaced by the weaker condition that
$(r_n)$ is pointedly $X$-computable, in the sense that $r_n$ is uniformly computable from
any extension of $X\restr_{r_n}$.
It is a simple exercise to extend our proof of this fact in order to replace the condition of pointed computation
with a non-uniform version of it, namely that
$K(r_n\ |\ \tau)=\bigo{1}$ for all $n$
and all $\tau\supseteq X\restr_{r_n}$. 
On the other hand, we showed that 
this condition is no longer sufficient for \ml randomness, if we merely require
that $r_n$ be computable from $X$. It would be interesting to refine this analysis and find
exactly what kind of computations are allowed of a sequence $(r_n)$ from an oracle $X$ such that
the \ml randomness of $X$ is equivalent to the segments $X\restr_{r_n}$ being incompressible. 


\end{document}